\documentclass[12pt, reqno]{amsart}

\usepackage{amsmath, amssymb, amsfonts, amsbsy, amsthm, latexsym}
\usepackage{graphicx}
\usepackage{verbatim}
\usepackage{color}

\newtheorem{theorem}{Theorem}[section]

\newtheorem{lemma}[theorem]{Lemma}

\theoremstyle{definition}
\newtheorem{definition}[theorem]{Definition}

\numberwithin{equation}{section}

\def\N{{\mathbb{N}}}

\numberwithin{equation}{section}

\usepackage{enumerate}

\def\0{{\bar{0}}}

\def\R{{\mathbb{R}}}

\begin{document}

\title[Sigma-Delta quantization for fusion frames]{High order low-bit Sigma-Delta quantization\\ for fusion frames}

 \author{Zhen Gao}
\address{Department of Mathematics, Vanderbilt University, Nashville, TN 37240, USA}
\email{zhen.gao.1@vanderbilt.edu}

\author{Felix Krahmer}
\address{Research Unit M15, Department of Mathematics, Technische Universit\"at M\"unchen, 85748 Munich, Germany}
\email{felix.krahmer@tum.de}

\author{Alexander M. Powell}
\address{Department of Mathematics, Vanderbilt University, Nashville, TN 37240, USA}
\email{alexander.m.powell@vanderbilt.edu}



\

\begin{abstract}
We construct high order low-bit Sigma-Delta $(\Sigma \Delta)$ quantizers for the vector-valued setting of fusion frames.
We prove that these $\Sigma \Delta$ quantizers can be stably implemented to quantize 
fusion frame measurements on subspaces $W_n$ using $\log_2( {\rm dim}(W_n)+1)$ bits per measurement. 
Signal reconstruction is performed using a version of Sobolev duals for fusion frames, and numerical experiments are given to validate the overall performance.
\end{abstract}

\maketitle

\section{Introduction} \label{intro-sec}

Fusion frames provide a mathematical setting for representing signals in terms of projections onto a redundant collection of closed subspaces.
Fusion frames were introduced in \cite{CK}  as a tool for data fusion, distributed processing, and sensor networks, e.g., see \cite{CKL,CKLR}.
In this work we consider the question of how to perform quantization, i.e., analog-to-digital conversion, on a collection of fusion frame measurements.

Our motivation comes from the stylized sensor network in \cite{JJAP}.  
Suppose that one seeks to measure a signal $x\in \R^d$ over a large environment using a collection of remotely dispersed sensors $s_n$ that are constrained by
limited power, limited computational resources, and limited ability to communicate. 
Each sensor is only able to make local measurements of the signal, and the goal is to communicate the local measurements to a distantly located base station where the signal $x$ can be accurately estimated.  The sensor network is modeled as a fusion frame and is physically constrained in the following manner:
\begin{itemize}
\item Each local measurement $y_n$ is a projection of $x$ onto a subspace $W_n$ associated to $s_n$.
\item Each sensor has knowledge of the proximities $W_n$ of a small number of nearby sensors.
\item Each sensor can communicate analog signals to a small number of nearby sensors.
\item Each sensor can transmit a {\em low-bit} signal to the distant base station.
\item The base station is relatively unconstrained in power and computational resources.
\end{itemize}

Mathematically, the above sensor network problem can be formulated as a quantization problem for fusion frames, e.g., \cite{JJAP}.  
Suppose that $\{W_n\}_{n=1}^N$ are subspaces of $\R^d$
and suppose that each $\mathcal{A}_n \subset W_n$ is a finite quantization alphabet.  Given $x\in \R^d$ and the orthogonal projections $y_n = P_{W_n}(x)$,
we seek an efficient algorithm for rounding the continuum-valued measurements $y_n\in W_n$ to finite-valued $q_n \in \mathcal{A}_n$.  This rounding process is called quantization and it provides a digital representation of $y_n$ through $q_n$.
Here, $q_n$ corresponds to the low-bit signal that the sensor $s_n$ transmits to the central base station.  We will focus on the case where the $q_n$ are computed sequentially and we allow the algorithm to be implemented with a limited amount of memory.  The memory variables correspond to the analog signals that sensors communicate to other
nearby sensors.  Finally, once the quantized measurements $\{q_n\}_{n=1}^N$ have been computed, we seek a reconstruction procedure for estimating $x$; this corresponds to 
the role of the base station.

We address the above problem with a new low-bit version of Sigma-Delta ($\Sigma \Delta$) quantization for fusion frames.  
Sigma-Delta quantization is a widely applicable class of algorithms for quantizing oversampled signal representations.  
Sigma-Delta quantization was introduced in \cite{IYM}, underwent extensive theoretical development in the engineering literature \cite{Gray},
and has been widely implemented in the circuitry of analog-to-digital converters \cite{NST}.  
Starting with \cite{DD}, the mathematical literature provided approximation-theoretic error bounds for Sigma-Delta quantization in a variety of settings, starting with
bandlimited sampling expansions \cite{DD,DGK, G,G-jams, GT,Y}. The best known constructions yield error bounds decaying exponentially in the bit budget \cite{G,DGK}, which is also the qualitative behavior that one encounters when quantizing Nyquist rate samples at high precision. The rate of this exponential decay, however, is provably slower for Sigma-Delta \cite{KW}.

Subsequently, Sigma-Delta was generalized to time-frequency representations \cite{Y-gabor} and finite frame expansions \cite{BPY}. As it turned out, the direct generalization to frames has significant limitations unless the frame under consideration has certain smoothness properties \cite{BodPau}. A first approach to overcome this obstacle was to recover signals using a specifically designed dual frame, the so-called Sobolev Dual \cite{BLPY,KSW}; this approach has also been implemented for compressed sensing \cite{GLPSY,KSY,FK}. Another class of dual frames that sometimes outperform Sobolev duals are the so-called beta duals \cite{CG}. In the context of compressed sensing, Sobolev duals have inspired a convex optimization approach for recovery \cite{SWY}, which has also been analyzed for certain structured measurements \cite{FKS}.

Since fusion frames employ vector-valued measurements, our approach in Definition \ref{ffsd-def} may be viewed as a vector-valued analogue of Sigma-Delta
quantization.  For perspective, we point out related work on $\Sigma\Delta$ quantization of finite frames with complex alphabets on a lattice \cite{BOT}, hexagonal $\Sigma\Delta$ modulators for power electronics \cite{LD05}, and dynamical systems motivated by error diffusion in digital halftoning \cite{AKMPST-2005, ANSTW-2015, ANST-2010}.

The work in \cite{JJAP} constructed and studied stable analogues of Sigma-Delta quantization in the setting of fusion frames.
The first order $\Sigma \Delta$ algorithms in \cite{JJAP} were stably implementable using very low-bit quantization alphabets.  Unfortunately, the higher order $\Sigma\Delta$ algorithms in \cite{JJAP} required large quantization alphabets for stability to hold.  
Stable high order $\Sigma\Delta$ algorithms are desirable since quantization error generally decreases as the order of a Sigma-Delta algorithm increases, e.g., \cite{DD}.
The main contribution of the current work is that we provide the first examples of stable high-order {\em low-bit} Sigma-Delta quantizers for fusion frames.

Our results achieve the following:
\begin{itemize}
\item We construct stable $r$th order fusion frame Sigma-Delta (FF$\Sigma\Delta$) algorithms with quantization alphabets that use $\log_2( {\rm dim}(W_n)+1)$-bits per subspace $W_n$, 
see Theorems \ref{stab-thm} and \ref{ord-r-thm}.  This resolves a question posed in \cite{JJAP}. 
\item We provide numerical examples to show that the FF$\Sigma\Delta$ algorithm performs well when implemented together with a version of Sobolev dual reconstruction.
\end{itemize}

\section{Fusion frames and quantization}

In this section, we provide background on fusion frames and quantization.

\subsection{Fusion frames}
Let $\{W_n\}_{n=1}^N$ be a collection of subspaces of $\R^d$ and let $\{c_n\}_{n=1}^N \subset (0,\infty)$ be positive scalar weights.
The collection $\{(W_n, c_n)\}_{n=1}^N$ is said to be a {\em fusion frame} for $\R^d$ with {\em fusion frame bounds} $0<A \leq B<\infty$ if
$$\forall x \in \R^d, \ \ \ A\|x\|^2 \leq \sum_{n=1}^N c_n^2 \|P_{W_n}(x)\|^2 \leq B\|x\|^2.$$
If the bounds $A,B$ are equal, then the fusion frame is said to be {\em tight}.
If $c_n=1$ for all $1 \leq n \leq N$, then the fusion frame is said to be {\em unweighted}.
Given a fusion frame, the associated {\em analysis operator} $T:\mathbb{R}^d \longrightarrow \bigoplus_{n=1}^N W_n$ is defined by 
$$T(x) = \{ c_n P_{W_n}(x) \}_{n=1}^N.$$
The problem of recovering a signal $x\in \R^d$ from fusion frame measurements $y_n = P_{W_n}(x)$ is equivalent to finding a left inverse to the analysis operator.  
There is a canonical choice of left inverse which can be described using the synthesis operator and the frame operator.

The adjoint $T^* \ : \bigoplus_{n=1}^N W_n  \longrightarrow \mathbb{R}^d$ of the analysis operator is called the {\em synthesis operator} and is defined by
$T^*(\{f_n\}_{n=1}^N) = \sum_{n=1}^N c_n f_n.$
The {\em fusion frame operator} $S: \R^d \to \R^d$ is defined by
$S(x) =( T^*T)(x)= \sum_{n=1}^N c_n^2 P_{W_n}(x).$
It is well-known \cite{CK} that $S$ is a positive self-adjoint operator.  Moreover, $L = S^{-1} T^*$ is a left inverse to $T$ since
${L}T=S^{-1} T^* T = S^{-1} S = I$.
This provides the following canonical reconstruction formula for recovering $x\in \R^d$ from fusion frame measurements $y_n =  P_{W_n}(x)$
\begin{equation}
\forall x \in \R^d, \ \ \ x = LT x = S^{-1} S x = \sum_{n=1}^N c_n^2 S^{-1}(y_n). \notag
\end{equation}
Although the canonical choice of left inverse $L = S^{-1} T^*$ is natural, 
other non-canonical left-inverses will be more suitable for the problem of reconstructing a signal $x$ from quantized measurements.

\subsection{Norms and direct sums}

The direct sum space $\bigoplus_{n=1}^N W_n$ arises naturally in the study of fusion frames.  In the  interest of maintaining simple notation,
we use the norm symbol $\| \cdot \|$ in different contexts throughout the paper to refer to norms on both Euclidean space and direct sum spaces, as well as operator norms on such spaces.

The following list summarizes different ways in which norm notation is used throughout the paper.
\begin{itemize}
\item If $x \in \R^d$ then $\|x\|$ denotes the Euclidian norm.  
\item If $H: \R^d \to \R^d$ then $\|H\| = \sup_{x \in \R^d} \frac{ \|Hx\|}{\|x\|}$.
\item If $w\in\bigoplus_{n=1}^N W_n$, then $\|w\| = \left( \sum_{n=1}^N \| w_n \|^2 \right)^{1/2}$ and $\|w\|_\infty = \sup_{1 \leq n \leq N} \|w_n\|.$
\item If $G: \bigoplus_{n=1}^N W_n \to \bigoplus_{n=1}^N W_n$ then $\|G\| = \sup_{w \in \bigoplus W_n} \frac{\|Gw\|}{\|w\|}$.
\item If $\mathcal{L}: \bigoplus_{n=1}^N W_n \to \R^d$ then $\|\mathcal{L}\| = \sup_{w \in \bigoplus W_n} \frac{\|\mathcal{L}w\|}{\|w\|}$.
\end{itemize}

\subsection{Quantization}

Let $x\in \R^d$ and suppose that $\{W_n\}_{n=1}^N$ are subspaces associated with a fusion frame for $\R^d$.
For each $1 \leq n \leq N$, let $\mathcal{A}_n \subset W_n$ be a finite set which we refer to as a {\em quantization alphabet}, and 
let $Q_n: W_n \to \mathcal{A}_n$ be an associated vector quantizer with the property that
\begin{equation} \label{vect-quant}
\forall w \in W_n,  \ \ \ \|w -  Q_n(w)\| =  \min_{q \in \mathcal{A}_n}  \|w -q\|.
\end{equation}

Memoryless quantization is the simplest approach to quantizing a set of fusion frame measurements $y_n = P_{W_n}(x), 1 \leq n \leq N$.
Memoryless quantization simply quantizes each $y_n$ to $q_n = Q_n(y_n)$.  See \cite{JJAP} for basic discussion on the performance of memoryless quantization for fusion frames.  This approach works well when the alphabets $\mathcal{A}_n$ are sufficiently fine and dense, and is also suitable when the subspaces are approximately orthogonal.  On the other hand, it is not very suitable for our sensor network problem which requires coarse low-bit alphabets and involves correlated subspaces $W_n$.   We will see that Sigma-Delta quantization is a preferable approach.

We will make use of the low-bit quantization alphabets provided by the following lemma.  These alphabets use $\log_2( {\rm dim}(W_n) +1)$ bits to quantize each subspace $W_n$.  For perspective,
in the scalar-valued setting, it is known that stable $\Sigma\Delta$ quantizers of arbitrary order can be implemented using a 1-bit quantization alphabet to quantize each scalar-valued sample, \cite{DD}.  The vector-valued alphabet in the following lemma provides a suitable low-bit analogue of this for fusion frames.

\begin{lemma} \label{simplex-lemma}
Let $W$ be an $m$-dimensional subspace of $\mathbb{R}^d$.  There exists a set $\mathcal{U}(W) =\{u_k\}_{k=1}^{m+1}$ in $W$ such
that $\sum_{j=1}^{m+1} u_j =0$, and each $u_j$ is unit-norm $\|u_j\|=1$, and
\begin{equation} \label{equi-angle}
\langle u_j, u_k \rangle = - \frac{1}{m}, \hbox{ for } j \neq k.
\end{equation}
Moreover, if $\theta_0 = cos^{-1} (\frac{1}{m})$, then for every $w \in W\backslash \{ 0\}$, there exists $k$ such that
\begin{equation} \label{theta-net}
angle(w, u_k) \leq \theta_0.
\end{equation}
\end{lemma}

For references associated to Lemma \ref{simplex-lemma}, see the discussion following Lemma 1 in \cite{JJAP}.

\section{Fusion frame Sigma-Delta quantization}

Throughout this section we shall assume that $\{W_n\}_{n=1}^\infty$ are subspaces of $\mathbb{R}^d$ and that each finite collection $\{W_n\}_{n=1}^N$ is an unweighted fusion frame for $\mathbb{R}^d$
when $N \geq d$.
We also assume that $\mathcal{A}_n =\mathcal{U}(W_n) \subset W_n$ is a set of $1+ {\rm dim} (W_n)$ vectors as in Lemma \ref{simplex-lemma},
and that $Q_n: W_n \to \mathcal{A}_n$ is a vector quantizer satisfying \eqref{vect-quant}.
Observe that by \eqref{vect-quant} and \eqref{theta-net} one has that for arbitrary $w\in W_n$ with $\|w\|=1$ 
\begin{align}\label{eq:mbq}
\langle Q_n(w), w\rangle = 1 - \tfrac{1}{2}\|Q_n(w) -w\|^2 = 1 - \tfrac{1}{2} \min_{q \in \mathcal{A}_n}  \|w -q\|^2 = \min_{q \in \mathcal{A}_n}  \langle q,w \rangle \geq \tfrac{1}{m}.
\end{align}

Given $x\in \R^d$, we shall investigate the following algorithm for quantizing fusion frame measurements $y_n = P_{W_n}(x)$.

\begin{definition}[Fusion frame Sigma-Delta algorithm]  \label{ffsd-def}
For each $n \geq 1$, fix operators $H_{n,j}: \R^d \to W_n$, $1 \leq j \leq L$.  
Initialize the state variables $v_0=v_{-1} = \cdots = v_{1-L} =0\in \mathbb{R}^d$.

The {\em fusion frame Sigma-Delta algorithm} (FF$\Sigma\Delta$) takes the measurements $y_n = P_{W_n}(x)$ as inputs and produces
quantized outputs $q_n \in \mathcal{A}_n$, $n \geq 1,$ by running the following iteration for $n\geq 1$
\begin{align}
q_n &= Q_n \left( y_n + \sum_{j=1}^{L} H_{n,j} (v_{n-j}) \right), \label{ffsd-eq1}\\
v_n &= y_n -q_n + \sum_{j=1}^{L} H_{n,j} (v_{n-j}).\label{ffsd-eq2}
\end{align}
\end{definition}

The algorithm  \eqref{ffsd-eq1}, \eqref{ffsd-eq2} may be applied to an infinite stream of inputs, but, in practice, the algorithm will usually be applied to a fusion frame of finite size and will terminate after finitely many steps.
The operators $H_{n,j}$ must be chosen carefully for the algorithm \eqref{ffsd-eq1}, \eqref{ffsd-eq2} to perform well.  We shall later focus on a specific choice of operators $H_{n,j}$ in 
Section \ref{rth-order-section},
but to understand its motivation, it is useful to first discuss reconstruction methods for the FF$\Sigma\Delta$ algorithm and to keep $H_{n,j}$ general for the moment.

The fusion frame Sigma-Delta algorithm must be coupled with a reconstruction procedure for recovering $x$ from the quantized measurements
$\{q_n\}$.  We consider the following simple reconstruction method that uses left inverses of fusion frame analysis operators.
At step $N$ of the FF$\Sigma\Delta$ algorithm, one has access to the quantized measurements $\{q_n\}_{n=1}^N$.  
Henceforth, $q\in \bigoplus_{n=1}^N {W}_n$ will denote the element of $\bigoplus_{n=1}^N W_n$ whose $n$th entry is $q_n \in \mathcal{A}_n \subset W_n$.
Since $\{W_n\}_{n=1}^N$ is a fusion frame with analysis operator $T=T_N$, let $\mathcal{L} = \mathcal{L}_N$ denote a left inverse of $T$, so that $\mathcal{L} Tx= x$ holds
for all $x\in \R^d$.  A specific choice of left inverse will be specified in Section \ref{rec-sec}, but for the current discussion let $\mathcal{L}$ be an arbitrary left inverse.
After step $N$ of the iteration \eqref{ffsd-eq1}, \eqref{ffsd-eq2}, we reconstruct the following $\widetilde{x}$ from $\{q_n\}_{n=1}^N$
\begin{equation} \label{lin-rec}
\widetilde{x} = \widetilde{x}_N = \mathcal{L} q.
\end{equation}

We now introduce notation that will be useful for describing the error $x- \widetilde{x}$ associated to \eqref{lin-rec}. 
Let $v \in \bigoplus_{n=1}^N {W}_n$ denote the element of $\bigoplus_{n=1}^N W_n$ whose $n$th entry is $v_n \in W_n$.
Let $\mathcal{I}_N:\bigoplus_{n=1}^N W_n \longrightarrow \bigoplus_{n=1}^N W_n$ denote the identity operator,
and let $\mathcal{H}_N:\bigoplus_{n=1}^N W_n \longrightarrow \bigoplus_{n=1}^N W_n$ denote the $N\times N$ block operator with entries
\begin{equation} \label{HN-entries}
\forall \ 1 \leq n,k \leq N, \ \ \ \ 
(\mathcal{H}_N)_{n,k} =
\begin{cases}
H_{n,n-k}, & \hbox{ if } 1 \leq n-k \leq L,  \\
0, & \hbox{otherwise.} 
\end{cases}
\end{equation}
Note that \eqref{ffsd-eq2} and \eqref{HN-entries} can be expressed in matrix form as $y-q=(\mathcal{I}_N-\mathcal{H}_N)v$.
Combining this and \eqref{lin-rec} allows the error  $x-\widetilde{x}$ to be expressed as
\begin{equation} \label{alg-err}
x - \widetilde{x} = \mathcal{L}_N(y - q) = \mathcal{L}_N(\mathcal{I}_N - \mathcal{H}_N)v.
\end{equation}

We aim to design the operators $\mathcal{H}_N$ in the quantization algorithm and the reconstruction operator $\mathcal{L}_N$ 
so that the error $\|x - \widetilde{x} \| = \|\mathcal{L}_N(\mathcal{I}_N - \mathcal{H}_N)v\|$ 
given by \eqref{alg-err} can be made quantifiably small.
We pursue the following design goals:
\begin{itemize}
\item Select $\mathcal{H}_N$ so that the iteration \eqref{ffsd-eq1}, \eqref{ffsd-eq2} satisfies a {\em stability condition} which controls the norm of the state variable sequence $v$.
\item Select $\mathcal{H}_N$ and $\mathcal{L}_N$ so that $\mathcal{L}_N(\mathcal{I}_N - \mathcal{H}_N)$ has small operator norm.  This can be decoupled into 
separate steps.  First, $\mathcal{H}_N$ is chosen to ensure that operator $\mathcal{I}_N - \mathcal{H}_N$ satisfies an {\em $r$th order condition} that expresses
$\mathcal{I}_N - \mathcal{H}_N$ in terms of a generalized $r$th order difference operator.
Secondly, $\mathcal{L}_N$ is chosen to be a {\em Sobolev left inverse} which is well-adapted to the operator $\mathcal{I}_N - \mathcal{H}_N$.
\end{itemize}
For the above points, Section \ref{stab-sec} discusses stability, Section \ref{rth-order-section} discusses the $r$th order property, and Section \ref{rec-sec} discusses reconstruction with Sobolev left inverses.

\section{Stability} \label{stab-sec}

The following theorem shows that the fusion frame $\Sigma\Delta$ algorithm is stable in the sense that control on the size of inputs $\|y_n\|$ ensures control on the size of state
variables $\|v_n\|$.  For perspective, the stable higher order fusion frame $\Sigma\Delta$ algorithm in \cite{JJAP} requires relatively large alphabets $\mathcal{A}_n$.

\begin{theorem} \label{stab-thm}
Let $\{W_n\}_{n=1}^N$ be subspaces of $\mathbb{R}^d$ with $d^\star = \max_{1 \leq n \leq N} {\rm dim}(W_n)$.
Suppose that a sequence $\{y_n\}_{n=1}^N$ with $y_n \in W_n$ is used as input to the algorithm \eqref{ffsd-eq1}, \eqref{ffsd-eq2}.

Suppose  that $0 < \delta < \frac{1}{d^\star}$, and let
$$\alpha_1 = \sqrt{\frac{1 - \frac{2\delta}{d^\star} + \delta^2}{1 - (\frac{1}{d^\star})^2} }
\ \ \ \hbox{ and } \ \ \ 
\alpha_2 = \frac{1}{2} \left( \left(\frac{1}{d^\star} - \delta \right) + \sqrt{\left( \frac{1}{d^\star} - \delta \right)^2+4} \thinspace \right).$$
Suppose that $\alpha= \sup_{1 \leq n \leq N} \sum_{j=1}^L \|H_{n,j}\|$ satisfies  $1 < \alpha \leq  \min \{ \alpha_1, \alpha_2 \},$ 
and let
$$C = \left(\frac{1}{d^\star} - \delta \right) \left( \frac{\alpha}{\alpha^2-1} \right).$$
If $\|y_n\| \leq \delta$ for all $1 \leq n \leq N$, then the state variables $v_n$ in \eqref{ffsd-eq1}, \eqref{ffsd-eq2} satisfy
$\|v_n\| \leq C$ for all $1 \leq n \leq N$.
\end{theorem}

\begin{proof}
\noindent {\em Step 1.} We begin by noting that the assumption $1 < \alpha \leq \min \{ \alpha_1, \alpha_2 \}$ is not vacuous.  The condition $1< \alpha_2$ directly follows from the assumption.
If $d^\star=1$ then $1 < \alpha_1 =\infty$ automatically holds.  
For $d^\star>1$, we rewrite 
\[
\alpha_1 = \sqrt{1 + \frac{(\frac{1}{d^\star} - \delta)^2}{1 - (\frac{1}{d^\star})^2} },
\]
which is strictly larger than $1$.\\

\noindent {\em Step 2.}  Next, we note that $C \geq 1$.  By the definition of $C$, it can be verified that $C\geq 1$ holds if and only if 
$$\alpha^2 - \left( \frac{1}{d^\star} - \delta \right) \alpha - 1 \leq 0.$$
It follows that $C\geq 1$ holds if and only if
$$\frac{\left( \frac{1}{d^\star} - \delta \right) - \sqrt{ \left( \frac{1}{d^\star} - \delta \right)^2 + 4}}{2}
\leq \alpha \leq \frac{\left( \frac{1}{d^\star} - \delta \right) + \sqrt{ \left( \frac{1}{d^\star} - \delta \right)^2 + 4}}{2}.$$
Since $\frac{\left( \frac{1}{d^\star} - \delta \right) - \sqrt{ \left( \frac{1}{d^\star} - \delta \right)^2 + 4}}{2} <0$ and 
$1<\alpha_2 = \frac{\left( \frac{1}{d^\star} - \delta \right) + \sqrt{ \left( \frac{1}{d^\star} - \delta \right)^2 + 4}}{2}$,
the assumption $1< \alpha \leq \alpha_2$ implies that $C\geq 1$, as required.\\

\noindent {\em Step 3.}  We will prove the theorem by induction.  The base case holds by the assumption that $v_0=v_{-1} = \cdots = v_{1-L} =0\in \mathbb{R}^d$.
For the inductive step, suppose that $n\geq 1$ and that $\|v_j\| \leq C$ holds for all $j \leq n-1$.

Let $z_n =  y_n+ \sum_{j=1}^{L} H_{n,j} (v_{n-j})$, so that $v_n = z_n -q_n$ with $q_n=Q_n(z_n)$.  If $z_n=0$ then $\|v_n\| = \| q_n\| = 1 \leq C$, as required.  So, it remains to consider $z_n \neq 0$.

When $z_n \neq 0$, let $\gamma_n =  \frac{\langle z_n, q_n \rangle}{\|z_n\|}$.  
Combining the definition of $d^{\star}$ and the fact that the quantizer $Q_n$ is scale invariant with \eqref{eq:mbq}, we obtain that $\gamma_n \geq  \frac{1}{d^\star}.$
Thus,
\begin{align}
\|v_n\|^2 &= \|z_n\|^2 +\|q_n\|^2 - 2 \|z_n\| \gamma_n \notag \\
&\leq \|z_n\|^2 - \frac{2}{d^\star} \|z_n\| +1. \label{vn-norm-bnd}
\end{align}

Since $\sum_{j=1}^L \|H_{n,j}\| \leq \alpha$, $\|y_n\| \leq \delta,$ and $\|v_{n-j}\| \leq C$, the definition of $z_n$ gives that
\begin{equation}
\| z_n \|  \leq  \|y_n\|+  \sum_{j=1}^{L} \|H_{n,j}\| \thinspace \|v_{n-j}\| \leq \delta +C\alpha. \label{zn-norm-bnd}
\end{equation}

Recall, we aim to show $\|v_n\| \leq C$.  Let $f(t) = t^2 - ( \frac{2}{d^\star}) t+1.$  By \eqref{vn-norm-bnd} and \eqref{zn-norm-bnd},  it suffices to prove
$$f([0,\alpha C+\delta]) \subseteq [0,C^2].$$
For that, we note that
$$\min \{ f(t) : t\in [0,\alpha C+\delta] \}  \geq f( {1}/{d^\star}) = 1 - ({1}/{d^\star})^2 \geq 0 .$$
and
\[\max \{ f(t) : t\in [0,\alpha C+\delta] \}  = \max \{ f(0), f(\alpha C+\delta) \} = \max \{ 1, f(\alpha C+\delta) \} \leq \max \{ C^2, f(\alpha C+\delta) \}.\]
Hence it only remains to show that $f(\alpha C+\delta) \leq C^2$.\\

\noindent {\em Step 4.}  Consider the polynomial $$p(x) =
(\alpha^2-1) x^2 + 2\alpha \left(\delta -   \frac{1}{d^\star} \right) x + \left(  1 -\frac{2 \delta}{d^\star}   + \delta^2\right).$$
Since $1<\alpha \leq \alpha_1$, it can be verified that the polynomial $p$ has real roots $r_1 \leq r_2$.  Since $\alpha >1$, one has that $p(x)\leq 0$ for all $x \in [ r_1,  r_2]$.
In particular, $p\left(\frac{r_1+r_2}{2} \right) \leq 0$.  Moreover, it can be checked that
$$\frac{r_1+r_2}{2} = \left(\frac{1}{d^\star} - \delta \right) \left( \frac{\alpha}{\alpha^2-1} \right) = C.$$
Thus, $p(C) \leq 0$.\\

\noindent {\em Step 5.}
Note that
\begin{align*}
f(\alpha C+\delta) &= (\alpha C+\delta)^2 -  \left( \frac{2}{d^{\star}} \right) (\alpha C+\delta) + 1\\
&= \alpha^2 C^2 + 2\alpha \left(\delta -   \frac{1}{d^\star} \right) C + \left(  1 -\frac{2 \delta}{d^\star}   + \delta^2\right)\\
&= p(C) +C^2.
\end{align*}
Since $p(C) \leq 0$ holds by Step 4, it follows that $f(\alpha C+\delta) \leq C^2$, as required.\\
\end{proof}

\section{$R$th order algorithms and feasibility} \label{rth-order-section}

Classical scalar-valued $r$th order Sigma-Delta quantization expresses coefficient quantization errors as an $r$th order difference of a bounded state variable, e.g., see \cite{DD, G}.
In this section we describe an analogue of this for the vector-valued setting of fusion frames.

Let $D_N:\bigoplus_{n=1}^N W_n \longrightarrow \bigoplus_{n=1}^N W_n$ be the $N\times N$ block operator defined by
\begin{equation} \label{DN-def}
\forall \ 1\leq n,k \leq N,  \ \ \ \ \ (D_N)_{n,k} =
\begin{cases}
I, & \hbox{ if } n=k,\\
-P_{W_n}, & \hbox{ if } n=k+1,\\
0, & \hbox{otherwise.}\\
\end{cases}
\end{equation}

\begin{definition}[$r$th order algorithm]
The fusion frame Sigma-Delta iteration \eqref{ffsd-eq1}, \eqref{ffsd-eq2} is an {\em $r$th order algorithm} if
for every $N \geq d$ there exist operators $G_N: \bigoplus_{n=1}^N W_n \longrightarrow \bigoplus_{n=1}^N W_n$
that satisfy
\begin{equation} \label{r-ord-def1}
\mathcal{I}_N - \mathcal{H}_N = (D_N)^rG_N,
\end{equation}
and
\begin{equation}\label{r-ord-def2}
\sup_N \|G_N\| < \infty.
\end{equation}
Moreover, given $\epsilon>0$, we say that $\{ (G_N, \mathcal{H}_N)\}_{N=d}^\infty$ is {\em $\epsilon$-feasible} if the operators $H_{n,j}$ that define $\mathcal{H}_N$ by \eqref{HN-entries} satisfy
\begin{equation} \label{epsilon-H-bnd}
\sup_{n \geq 1} \sum_{j=1}^L \|H_{n,j} \| \leq 1 + \epsilon.
\end{equation}
\end{definition}

The $r$th order condition \eqref{r-ord-def1} should be compared with the scalar-valued analogue in equation (4.2) in \cite{G}, cf. \cite{DGK}.
For perspective, the condition \eqref{epsilon-H-bnd} ensures that the stability result from Theorem \ref{stab-thm} can be used. 
The $r$th order conditions \eqref{r-ord-def1}, \eqref{r-ord-def2} will later be used in Section \ref{rec-sec} to provide control on the quantization error $\|x - \widetilde{x}\|$.

We now show that it is possible to select $\mathcal{H}_N$ so that the low-bit fusion frame Sigma-Delta algorithm in \eqref{r-ord-def1}, \eqref{r-ord-def2} 
is $r$th order and $\epsilon$-feasible with small $\epsilon>0$.

We make use of the following sequences $n_j, d_j, h$ defined in \cite{G}. The constructions have subsequently been improved in \cite{K, DGK}, but we will work with the (suboptimal) first construction, as it allows for closed form expressions.
Given $r \in \N$, define the index set $\mathbb{N}_r = \N \cap [1,r]$.
Let $r, \sigma \in \N$ be fixed and define the sequences $\{n_j\}_{j=1}^r$ and $\{d_j\}_{j=1}^r$ by
\begin{align} \label{nj-dj-defs}
n_j = \sigma (j-1)^2+1 \ \ \ \hbox{ and } \ \ \ d_j = \prod_{i \in \mathbb{N}_r  \backslash \{j\}} \frac{n_i}{n_i-n_j}.
\end{align}
Next, define $h \in \ell^1(\N)$ by 
\begin{equation} \label{h-def}
h = \sum_{j=1}^r d_j \delta_{n_j},
\end{equation}
where $\delta_n \in \ell^1(\N)$ is defined by $\delta_n(j) = 1$ if $j=n$, and $\delta_n(j) = 0$ if $j\neq n$.
We will later use the property, proven in \cite{G}, that $h$ satisfies
\begin{equation}\label{h-ell-one}
\|h\|_{\ell^1} < \cosh (\pi \sigma^{-1/2}).
\end{equation}

Let $L=n_r$ and define the $N\times N$ block operator $\mathcal{H}_N$ using \eqref{HN-entries} with
$1 \leq n \leq N$ and $1\leq j \leq L$ and
\begin{equation} \label{H-def}
H_{n,j} = 
\begin{cases}
h_j P_{W_n}P_{W_{n-1}} \cdots P_{W_{n-j+1}}, &\hbox{ if } n>j,\\
0, & \hbox{otherwise.}
\end{cases}
\end{equation}

In the following result, we prove that the fusion frame Sigma-Delta algorithm with operators \eqref{H-def} is $r$th order and $\epsilon$-feasible.
\begin{theorem} \label{ord-r-thm}
Fix $r\geq 2$.  Given $\epsilon>0$, if $\sigma\in \N$ is sufficiently large and if the operators $H_{n,j}: \bigoplus_{n=1}^N W_n \longrightarrow \bigoplus_{n=1}^N W_n$ are defined by 
\eqref{nj-dj-defs}, \eqref{h-def}, \eqref{H-def}, 
then the fusion frame Sigma-Delta algorithm \eqref{ffsd-eq1}, \eqref{ffsd-eq2} is an $r$th order algorithm and is $\epsilon$-feasible.
\end{theorem}

The proof of Theorem \ref{ord-r-thm} is given in Section \ref{pf-ord-r-thm-sec}.

\section{Background lemmas}

In this section, we collect background lemmas that are needed in the proof of Theorem \ref{ord-r-thm}.
The following result provides a formula for the entries of the block operator $(D_N)^{-r}$.
\begin{lemma} \label{DN-inv-r-lem}
Fix $r \geq 1$.  If $D_N$ is the $N\times N$ block operator defined by \eqref{DN-def} then $D_N$ is invertible and $D^{-r}_N$ satisfies
\begin{align}
(D_N^{-r})_{i,j} = \begin{cases} {r+i-j-1 \choose r-1}(P_{W_i}P_{W_{i-1}}\dots P_{W_{j+1}} )  \ \ &\text{if } i>j\\
I &\text{if } i=j\\
0 &\text{if }i<j.
\end{cases}
\label{DN-inv-r-eq}
\end{align}
\end{lemma}

\begin{proof}
The proof proceeds by induction.  For the base case $r=1$, a direct computation shows that
\begin{align*}
\forall \ 1 \leq i, j \leq N,  \ \ \ (D_N^{-1})_{i,j} =
\begin{cases}
P_{W_i} P_{W_{i-1}} \cdots P_{W_{j+1}}  \ \ &\text{if } i>j\\
I &\text{if } i=j\\
0 &\text{if }i<j.
\end{cases}
\end{align*}

For the inductive step, suppose that \eqref{DN-inv-r-eq} holds.  
Using $(D_N^{-(r+1)})_{i,j} = \sum_{k = 1}^N (D_N^{-r})_{i,k} (D_N^{-1})_{k,j},$ 
shows that $(D_N^{-(r+1)})_{i,i} =I$, and that if $i<j$ then $(D_N^{-(r+1)})_{i,j} =0$.
If $i>j$, then
\begin{align} \label{DN-i-ge-j}
(D_N^{-(r+1)})_{i,j} =  \sum_{k = j}^i (D_N^{-r})_{i,k} (D_N^{-1})_{k,j}
=  \left(\sum_{k = j}^i {r+i-k-1 \choose r-1}\right)(P_{W_i}P_{W_{i-1}}\dots P_{W_{j+1}} ).
\end{align}
The combinatorial identity $\sum_{k=0}^{i-j} {r-1+k \choose k} = \sum_{k=0}^{i-j}{r-1+k \choose r-1} = {r + i -j \choose r}$, e.g., see page 1617 in \cite{G}, shows that
$\sum_{k = j}^i {r+i-k-1 \choose r-1} = \sum_{k=j}^i {r-1+i-k \choose i-k} = \sum_{k=0}^{i-j} {r-1+k \choose k} = {r + i -j \choose r}.$
In particular, \eqref{DN-i-ge-j} reduces to $(D_N^{-(r+1)})_{i,j} = {r + i -j \choose r} (P_{W_i}P_{W_{i-1}}\dots P_{W_{j+1}} )$ when $i>j$.
\end{proof}

\begin{lemma}  \label{sinan-lemma}
Fix $\sigma \in \mathbb N$, $r \in \mathbb N$, and define $\{h_l\}_{l\in \N}$, $\{n_j\}_{j=1}^r$, $\{d_j\}_{j=1}^r$ by \eqref{nj-dj-defs} and \eqref{h-def}.
If $n\geq n_r - r + 1$, then
\begin{align}  \label{sinan-lemma-eq}
{r+n-1 \choose r-1} = \sum_{l=1}^n {r+n-1 -l\choose r-1} h_l.
\end{align}
\end{lemma}

\noindent {\em Sketch of Proof.} The result is contained in the proof of Proposition 6.1 in \cite{G}.  We provide a brief summary since \cite{G} proves a more general result.

First, note that $\{ n_j\}_{j=1}^r$ is an increasing sequence of strictly positive, distinct integers, which satisfies the requirements of Proposition 6.1 in \cite{G}.
The final sentence in step (i) of the proof Proposition 6.1 in \cite{G} shows that
\begin{equation*}
{n+r-1 \choose r-1} - \sum_{n_j \leq n} d_j {n -n_j +r-1 \choose r-1}=  g_n,
\end{equation*}
where $g_n$ is given by (6.1) in \cite{G}.  Moreover, the first two sentences in step (ii) in the proof of Proposition 6.1 in \cite{G} give
that $g_n= \left( \prod_{i=1}^r n_i \right) G(n)$ where $G(n)=0$ whenever $n\geq n_r-r+1$.
Finally, recalling the definition $h_l$ in \eqref{h-def} gives the desired conclusion.
\qed

\section{Proof of Theorem \ref{ord-r-thm}} \label{pf-ord-r-thm-sec}

In this section we prove Theorem \ref{ord-r-thm}.\\

\noindent {\em Step 1.}  We first show that the operators $H_{n,j} : \bigoplus ^N_{n=1} W_n \to \bigoplus ^N_{n=1} W_n$ defined by \eqref{nj-dj-defs}, \eqref{h-def}, \eqref{H-def}  satisfy \eqref{epsilon-H-bnd} when $\sigma \in \N$ is sufficiently large. 

Note that $f(x) = \cosh(x)$ is decreasing on $(0, \infty)$ and $\lim_{x \to 0+} \cosh(x) = 1$.
Given $\epsilon>0$, it follows from \eqref{h-ell-one} that there exists $N_0=N_0(\epsilon)$ so that $\sigma > N_0$ implies 
\begin{equation} \label{h-one-plus-eps}
\|h\|_{\ell^1} < \cosh(\pi \sigma^{-1/2}) < 1 + \epsilon.
\end{equation}
By \eqref{H-def} we have
\begin{align} \label{Hhbnd}
\sup_{n\geq 1}\sum^L_{j=1} \| H_{n,j}\| = \sup_{n \geq 1}\sum^L_{j=1} \|h_j P_{W_n} P_{W_{n-1}} \dots P_{W_{n-j+1}}\|
\leq \sup_{n \geq 1}\sum^L_{j=1} |h_j| \leq \|h\|_{\ell^1} < 1+\epsilon .
\end{align}

\noindent {\em Step 2.}  Define the $N \times N$ block operator $G_N = (D_N^{-r})(\mathcal{I}_N - \mathcal{H}_N)$.  Using \eqref{HN-entries}, \eqref{H-def}, Lemma \ref{DN-inv-r-lem},
and $(G_N)_{i,j} =\sum_{k=1}^N (D_N^{-r})_{i,k} (\mathcal{I}_N - \mathcal{H}_N)_{k,j}$ it can be shown that
\begin{equation}
(G_N)_{i,j} = 
\begin{cases} 
\left({r+i-j-1 \choose r-1} - \sum_{l=1}^{i-j} {r+i-j-l-1 \choose r-1}h_l
\right)\ 
(P_{W_i} P_{W_{i-1}} \cdots P_{W_{j+1}}) &\text{ if } i > j, \\ 
I &\text{ if } i = j, \\ 
0 &\text { if } i < j. \\
\end{cases}
\end{equation}

Let $K=n_r-r+1$.  Lemma \ref{sinan-lemma} shows that if $i-j\geq K$ then ${r+i-j-1 \choose r-1} = \sum_{l=1}^{i-j} {r+i-j-l-1 \choose r-1}h_l$.
This shows that $G_N$ is banded and satisfies
\begin{equation} \label{GN-banded-entries}
(G_N)_{i,j} = 
\begin{cases} 
\left({r+i-j-1 \choose r-1} - \sum_{l=1}^{i-j} {r+i-j-l-1 \choose r-1}h_l
\right)\ 
(P_{W_i} P_{W_{i-1}} \cdots P_{W_{j+1}}) &\text{ if } K> i -j>0, \\ 
I &\text{ if } i = j, \\ 
0 &\text { otherwise.} \\
\end{cases}
\end{equation}
\vspace{.1in}

\noindent {\em Step 3.}  Recall that $K = n_r-r+1$ and let $M_r = {r+K-2 \choose r-1}.$ We next show that
if $1 \leq i,j \leq N$ and $0<i-j<K$ then
\begin{equation} \label{GN-op-bnd-unif}
\| (G_N)_{i,j} \| \leq (2+\epsilon)M_r.
\end{equation}

Since ${r+m-2 \choose r-1}$ increases as $m$ increases, it follows that if $0<i-j<K$ then
${r+i-j-1 \choose r-1} \leq {r+K-2 \choose r-1}=M_r $.  Likewise, if $1 \leq l \leq i-j<K$ then${r+i-j-l-1 \choose r-1} \leq {r+K-2 \choose r-1}=M_r $.
Also, recall that by \eqref{h-one-plus-eps} we have $\|h \|_{\ell^1} < 1 + \epsilon$.  
So, if $0<i-j<K$ then \eqref{GN-banded-entries} implies that
\begin{align*}
\|(G_N)_{i,j} \| & \leq \left| {r+i-j-1 \choose r-1} \right| + \left| \sum_{l=1}^{i-j} {r+i-j-l-1 \choose r-1}h_l \right| \\
&\leq M_r + M_r   \sum_{l=1}^{i-j} | h_l| \leq M_r + M_r  \|h\|_{\ell^1} \leq (2+\epsilon)M_r.
\end{align*}

\noindent {\em Step 4.}  
Next, we prove that
\begin{equation}
\sup_{N>K} \|G_N\|  \leq (2+\epsilon) M_r K< \infty.
\end{equation}
Suppose that $v \in \bigoplus ^N_{n=1} W_n$ satisfies $\|v\|_2=1$.
By \eqref{GN-banded-entries}, it follows that $G_Nv$ satisfies
\begin{equation} \label{GNv-entries}
\forall \ 1 \leq n \leq N, \ \ \ (G_Nv)_n = 
\begin{cases}
v_n + \sum_{k=1}^{n-1} G_{n,k} v_k & \hbox{ if } 1 \leq n \leq K,\\
v_n + \sum_{k=n+1-K}^{n-1} G_{n,k} v_k & \hbox{ if } K+1 \leq n \leq N.
\end{cases}
\end{equation}
Using  \eqref{GN-op-bnd-unif}, \eqref{GNv-entries}, and noting that $(2+\epsilon)M_r \geq 1$, gives
\begin{align}
\|G_N v\|^2 &= \sum_{n=1}^K \|v_n +\sum_{k=1}^{n-1} G_{n,k} v_k \|^2 + \sum_{n=K+1}^N \|v_n + \sum_{k=n+1-K}^{n-1} G_{n,k} v_k\|^2 \notag \\
& \leq (2+\epsilon)^2 M_r^2 \sum_{n=1}^K \left( \|v_n\| +\sum_{k=1}^{n-1} \|v_k \| \right)^2 + (2+\epsilon)^2 M_r^2 \sum_{n=K+1}^N \left( \|v_n\| + \sum_{k=n+1-K}^{n-1} \|v_k\| \right)^2 \notag \\
&= (2+\epsilon)^2 M_r^2 \left[  \sum_{n=1}^K \left( \sum_{k=1}^{n} \|v_k \| \right)^2 + \sum_{n=K+1}^N \left( \sum_{k=n+1-K}^{n} \|v_k\| \right)^2 \right]. \label{GNv-eq}
\end{align}
Applying the Cauchy-Schwarz inequality to \eqref{GNv-eq} gives
\begin{align}
\|G_N v\|^2 &\leq (2+\epsilon)^2 M_r^2 \left[  \sum_{n=1}^K n \sum_{k=1}^{n} \|v_k \|^2 + \sum_{n=K+1}^N K  \sum_{k=n+1-K}^{n} \|v_k\|^2 \right]\notag \\
&\leq (2+\epsilon)^2 M_r^2 K \left[  \sum_{n=1}^K \sum_{k=1}^{n} \|v_k \|^2 + \sum_{n=K+1}^N  \sum_{k=n+1-K}^{n} \|v_k\|^2 \right]. \label{GNv-eq-afterCS}
\end{align}
Next, a computation shows that
\begin{equation} \label{reindexed-sum-eq}
\sum_{n=K+1}^N  \sum_{k=n+1-K}^{n} \|v_k\|^2 = \sum_{n=1}^K  \sum_{k=n+1}^{N-K+n} \|v_{k}\|^2.
\end{equation}
Combining \eqref{GNv-eq-afterCS} and \eqref{reindexed-sum-eq} completes the proof
\begin{align*}
\|G_N v\|^2 &\leq (2+\epsilon)^2 M_r^2 K  \ \sum_{n=1}^K \sum_{k=1}^{N-K+n} \|v_k \|^2  \leq (2+\epsilon)^2 M_r^2 K^2 \sum_{k=1}^N \|v_k \|^2 =(2+\epsilon)^2 M_r^2 K^2.
\end{align*}
\qed

\section{Reconstruction and error bounds} \label{rec-sec}

In this section, we describe the choice of left inverse $\mathcal{L}$ used for the reconstruction step \eqref{lin-rec}.
Combining the error expression \eqref{alg-err} with the $r$th order condition \eqref{r-ord-def1} gives
\begin{align} \label{err-LDGv}
x - \widetilde{x} = \mathcal{L} (\mathcal{I}_N - \mathcal{H}_N)v = \mathcal{L} D_N^{r} G_N v.
\end{align}
If $T:\R^d \to \bigoplus_{n=1}^N W_n$ is the analysis operator of the unweighted fusion frame $\{W_n\}_{n=1}^N$,
we seek a left inverse $\mathcal{L}:\bigoplus_{n=1}^N W_n \to \R^d$ that satisfies $\mathcal{L} T = I$ and for which the quantization error  in \eqref{err-LDGv} is small.

By the stability result in Theorem \ref{stab-thm}, the state variable $v$ satisfies $\|v\|_\infty  \leq C<\infty$ and $\|v\| \leq \sqrt{N} \|v\|_\infty$. 
Also, the $r$th order condition \eqref{r-ord-def2} ensures that  $C^\prime=\sup_N \|G_N\| <\infty$ is finite.
So
\begin{align}
\|x - \widetilde{x}\| & \leq \| \mathcal{L} D_N^{r} G_N \| \thinspace \| v\| \notag \\
& \leq \| \mathcal{L} D_N^{r}\| \thinspace \| G_N \| \thinspace \| v\| \notag  \\
& \leq C^\prime  \sqrt{N} \thinspace \| \mathcal{L} D_N^{r}\|\| v\|_\infty \notag \\
& \leq C^\prime C \sqrt{N}  \| \mathcal{L} D_N^{r}\|. \label{err-LDr}
\end{align}

In view of \eqref{err-LDr}, the work in \cite{JJAP} selected $\mathcal{L}$ as a left inverse to $T$ that makes $\| \mathcal{L} D_N^{r}\|$ small.
The {\em $r$th order Sobolev left inverse} is defined by
\begin{equation} \label{sob-inv}
\mathcal{L}_{r, Sob} = ((D_N^{-r}T)^*D_N^{-r}T)^{-1}(D_N^{-r}T)^*D_N^{-r}.
\end{equation}
It is easily verified that $\mathcal{L}_{r, Sob} \thinspace T =I$; see \cite{JJAP} for further discussion of Sobolev duals in the setting of fusion frames.

In general, it can be difficult to bound the operator norm $\| \mathcal{L} D_N^{r}\|$ in \eqref{err-LDr}.  It would be interesting to find quantitative bounds on 
$\| \mathcal{L} D_N^{r}\|$ when $\mathcal{L}$ is the Sobolev left inverse for nicely structured classes of deterministic or random fusion frames.   
For perspective,  \cite{BLPY, FK, GLPSY, KSW, SWY} provides analogous bounds for the scalar-valued setting of frames, 
and \cite{JJAP} contains examples for fusion frames when $\mathcal{L}$ is the canonical left inverse.

\section{Numerical experiments}

This section contains two numerical examples which illustrate the performance of the low-bit fusion frame Sigma-Delta algorithm.
For each $N\geq 3$, define the unit-vectors $\{\varphi^N_n\}_{n=1}^N \subset \R^3$ by
$$\varphi_n^N =  \begin{pmatrix} \frac{1}{\sqrt{3}}, & \sqrt{\frac{2}{3}}\cos(\frac{2\pi n}{N}), &  \sqrt{\frac{2}{3}}\sin(\frac{2\pi n}{N}) \end{pmatrix},$$
and define the unweighted fusion frame $\mathcal{W}_N = \{W^N_n\}_{n=1}^N$ by
$$W_n^N = \{ x \in \R^3 : \langle x, \varphi_n^N \rangle = 0 \}.$$

For each fixed $N \geq 3$,  $\{W_n^N\}_{n=1}^N$ is an unweighted tight fusion frame for $\mathbb R^3$ with fusion frame bound $A=A_N= \frac{2N}{3}$, e.g., see Example 1 in \cite{JJAP}.
Moreover, the vectors $e_{1,n}^N = \left(0, \sin(\frac{2\pi n}{N}), -\cos(\frac{2\pi n}{N})\right)$ and 
$e_{2,n}^N = \sqrt{\frac{1}{3}} \left(-\sqrt{2} , \cos(\frac{2\pi n}{N}), \sin(\frac{2\pi n}{N})\right)$ form an orthonormal basis for $W_n^N$.

Let $\mathcal A_n^N \subset W_n^N$ be the low-bit quantization alphabet given by Lemma \ref{simplex-lemma}.  Since ${\rm dim}(W_n^N) = 2$,
each alphabet $\mathcal A_n^N$ contains 3 elements, and can be defined by
$$\mathcal A_n^N = \left\{e_{1,n}^N, \left(-\frac{1}{2} e_{1,n}^N + \frac{\sqrt{3}}{2} e_{2,n}^N\right), \left(-\frac{1}{2}e_{1,n}^N - \frac{\sqrt{3}}{2} e_{2,n}^N\right)\right\}.$$
Let $Q_n$ be a vector quantizer associated to $\mathcal{A}_n^N$ by \eqref{vect-quant}.

\subsection*{Example 1 (second order algorithm)}
This example considers the performance of the second order fusion frame Sigma-Delta algorithm.
 By Theorems \ref{stab-thm} and \ref{ord-r-thm} we can choose appropriate $\sigma$ and $h$, as in Section \ref{rth-order-section}, to ensure
 that the algorithm \eqref{ffsd-eq1}, \eqref{ffsd-eq2} is stable and second order. 
 In Theorem \ref{stab-thm}, let $\delta = 0.1$, so that $\alpha_1 \approx 1.1015$ and $\alpha_2 \approx 1.2198$ allows us to pick $\alpha = 1.101$. 
Using  \eqref{h-ell-one} and \eqref{Hhbnd}, the condition $\sup_{n \geq 1} \sum_{j=1}^L \| H_{n,j}\| \leq \|h\|_{\ell^1}< \alpha$ will be satisfied if
 $\pi \sigma^{-1/2} = \cosh^{-1}(\alpha) = \ln(\alpha + \sqrt{\alpha^2 -1}) \leq 0.4458$, which occurs when $\sigma \geq 49.67$.  We pick $\sigma = 50$, so  that
 \eqref{nj-dj-defs} gives $n_1 = 1, n_2 = 51$, and \eqref{h-def} gives
\begin{align}
h_j = 
\begin{cases}
\frac{n_2}{n_2 - n_1} = \frac{51}{50} \ \ \ \ \ &\text{if } j = n_1,\\
\frac{n_1}{n_1 - n_2} = -\frac{1}{50} \ \ \ &\text{if } j = n_2,\\
0  \ \ \ \  \ \ \ \ \ &\text{otherwise}.
\end{cases}
\end{align}
The second order low-bit $\Sigma\Delta$ quantization algorithm takes the following form
\begin{align}
&q^N_n = Q_n\left(y^N_n + h_1 P_{W^N_n}(v^N_{n-1}) + h_{n_2} \prod_{k=0}^{n_2-1} P_{W^N_{n - k}}(v^N_{n-n_2})\right), \\
&v^N_n = y^N_n - q^N_n+ \left( h_1 P_{W^N_n}(v^N_{n-1}) + h_{n_2} \prod_{k=0}^{n_2-1} P_{W^N_{n - k}}(v^N_{n-n_2}) \right).
\end{align}

Let $x = (\frac{1}{25}, \frac{\pi}{57}, \frac{1}{2\sqrt{57}})$ 
and define the fusion frame measurements by $y_n^N = P_{W_n^N}(x)$. 
Note that $\|y_n\| \leq \|x\| \leq \delta$.
Run the second order low-bit fusion frame Sigma-Delta algorithm with inputs $\{y_n^N\}_{n=1}^N$, to obtain the quantized outputs $q^N = \{q^N_n\}_{n=1}^N$.

Let $T_N$ be the analysis operator for the unweighted fusion frame $\{W_n^N\}_{n=1}^N$.  The canonical left inverse of $T_N$ is $\mathcal{L}_N = S_N^{-1} T_N^* = (A_N^{-1}I) T_N^* = A_N^{-1}T_N^*$.  
Since the fusion frame  is tight with bound $A_N = 2N/3$,
it follows that $\mathcal{L}_N = \frac{3}{2N} T_N^*$, e.g., \cite{JJAP}. 
Also let $\mathcal{L}^N_{2, Sob}$ be the second order Sobolev left inverse of $T_N$, as defined in \eqref{sob-inv}. 
Consider the following two different methods of reconstructing a signal from $q^N$
$$\widetilde{x}_N = \mathcal{L}_N(q^N) \ \ \ \hbox{ and } \ \ \ \widetilde{x}_{N,Sob} = \mathcal{L}^N_{2,Sob}(q^N).$$

\begin{figure} 
\centering
\includegraphics[scale=0.3]{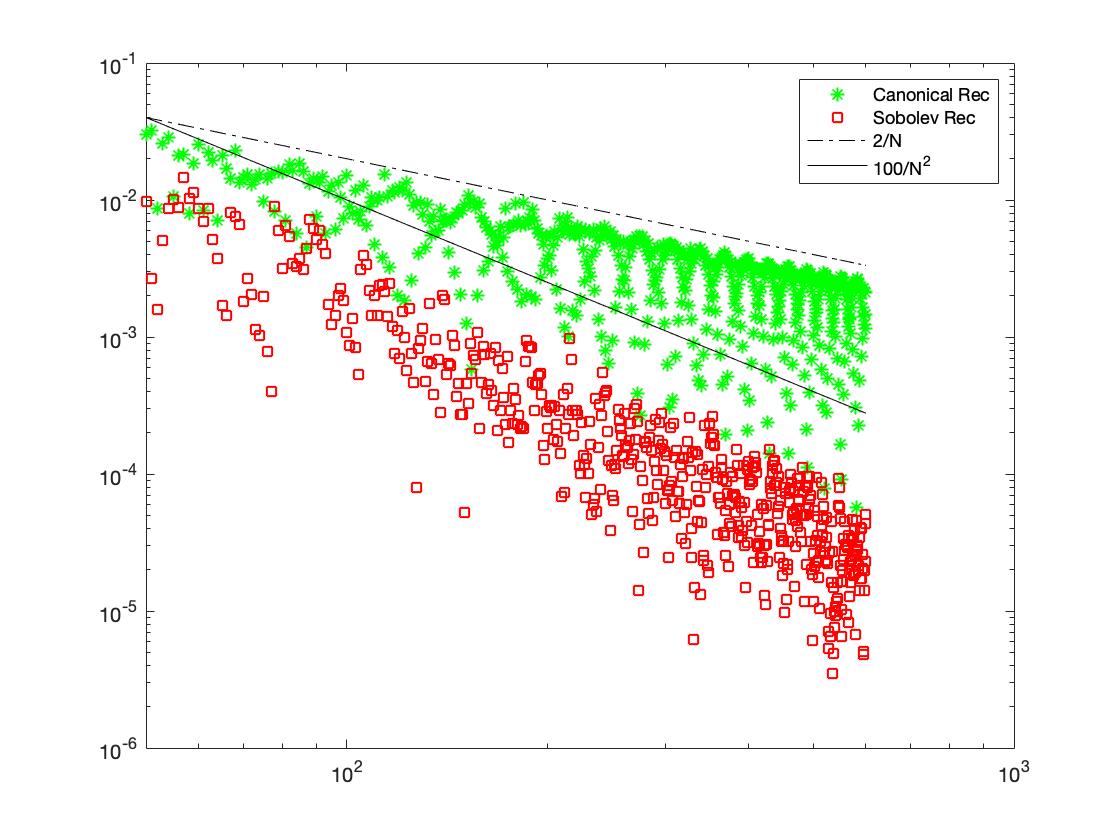}
\caption[]
{\label{fig1} Error for the second order algorithm in Example 1.}
\end{figure}
Figure \ref{fig1} shows log-log plots of $\|x - \widetilde{x}_N\|$ and $\|x - \widetilde{x}_{N,Sob}\|$ against $N$. For comparison, log-log plots of $2/N$ and $100/N^2$ against $N$ are also given.

\subsection*{Example 2 (third order algorithm)}

We consider the same experiment as in Example 1, except that we use an algorithm of order $r = 3$.

We again use the parameters $\delta =0.1$ and $\sigma =50$.  Using \eqref{nj-dj-defs} with $\sigma=50$ and $r=3$ gives $n_1 = 1, n_2 = 51, n_3 = 201$ and
\begin{align}
h_j = 
\begin{cases}
\frac{n_2n_3}{(n_2 - n_1)(n_3 - n_1)}  \ \ \ \ \ &\text{if } j = n_1,\\
\frac{n_1n_3}{(n_1 - n_2)(n_3 - n_2)}  \ \ \ &\text{if } j = n_2, \\
\frac{n_1n_2}{(n_1 - n_3)(n_2 - n_3)}  \ \ \ &\text{if } j = n_3, \\
0  \ \ \ \  \ \ \ \ \ &\text{otherwise}.
\end{cases}
\end{align}
The third order low-bit fusion frame Sigma-Delta quantization algorithm takes the following form
\begin{align}
&q^N_n = Q_n\left(y^N_n + h_1 P_{W^N_n}(v^N_{n-1}) + h_{n_2} \prod_{k=0}^{n_2-1} P_{W^N_{n - k}}(v^N_{n-n_2}) + h_{n_3} \prod_{k=0}^{n_3-1} P_{W^N_{n - k}}(v^N_{n-n_3})\right), \\
&v^N_n = y^N_n - q^N_n+ \left( h_1 P_{W^N_n}(v^N_{n-1}) + h_{n_2} \prod_{k=0}^{n_2-1} P_{W^N_{n - k}}(v^N_{n-n_2}) + h_{n_3} \prod_{k=0}^{n_3-1} P_{W^N_{n - k}}(v^N_{n-n_3}) \right).
\end{align}
Let $x = (\frac{1}{25}, \frac{\pi}{57}, \frac{1}{2\sqrt{57}})$, and use the third order fusion frame Sigma-Delta algorithm 
with inputs $\{y_n^N\}_{n=1}^N$ to obtain the quantized outputs $q^N = \{q^N_n\}_{n=1}^N$.
For the reconstruction step, let $\mathcal{L}_N = \frac{3}{2N} T^*_N$ be the canonical left inverse of $T_N$ and let $\mathcal{L}_{3,Sob}^N$ be the third order Sobolev left inverse of $T_N$, as defined in \eqref{sob-inv}.   We consider the following two different methods of reconstructing a signal from $q^N$
$$\widetilde{x}_N = \mathcal{L}_N(q^N) \ \ \ \hbox{ and } \ \ \ \widetilde{x}_{N,Sob} = \mathcal{L}^N_{3,Sob}(q^N).$$

\begin{figure}
\centering
\includegraphics[scale=0.3]{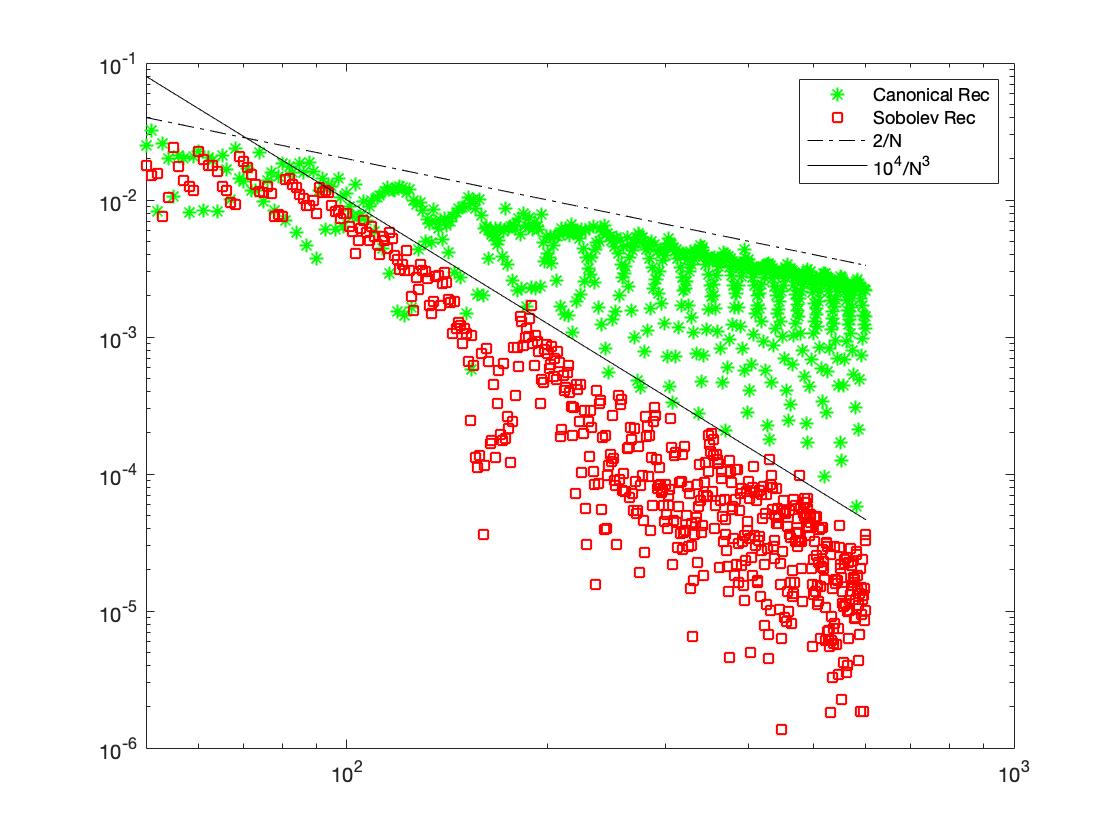}
\caption[]
{\label{fig2} Error for the third order algorithm in Example 2.}
\end{figure}

Figure 2 shows log-log plots of $\|x - \widetilde{x}_N\|$ and $\|x - \widetilde{x}_{N,Sob}\|$ against $N$. For comparison, log-log plots of $2/N$ and $10^4/N^3$ against $N$ are also given.

\section{Outlook}
In this paper we have discussed higher order Sigma-Delta modulators for fusion frame measurements and proved their stability. As for finite frames, the reconstruction accuracy of such approaches will depend on the fusion frame under consideration. In particular, we expect that for certain adversarial fusion frame constructions, only very slow error decay can be achieved. 

On the other hand, our numerical experiments in the previous section show that for certain deterministic fusion frames the error decays polynomially of an order that corresponds to the order of the Sigma-Delta scheme.  
For random frames, such error bounds have been established with high probability \cite{GLPSY, KSY,FK}. These result have important implications
for compressed sensing with random measurement matrices.  Given that the theory of compressed sensing generalizes to the setting of fusion frames \cite{BKR},
and there exists analysis of random fusion frames which parallels the restricted isometry property \cite{Bod}; it will be interesting to understand if the aforementioned results  generalize to the stable low-bit $r$th order 
fusion frame Sigma-Delta algorithms discussed in this paper, or whether modifications are necessary. The crucial  quantity to estimate is the last factor in \eqref{err-LDr} for the Sobolev dual of a random fusion frame. In any case, we expect that the stability analysis provided in this paper will be of crucial importance even in the latter case.

\section*{Acknowledgements}
The authors thank Keaton Hamm and Jiayi Jiang for valuable discussions related to the material.
F.~Krahmer was suported in part by the German Science Foundation in the context of the Emmy Noether junior research group KR4512/1-1.
A.~Powell was supported in part by NSF DMS Grant 1521749, and gratefully acknowledges the Academia Sinica Institute of Mathematics (Taipei, Taiwan) for their hospitality and support.

\bibliographystyle{amsplain}

\end{document}